\documentclass[letterpaper, 10 pt, conference]{ieeeconf}  

\IEEEoverridecommandlockouts 

\title{\LARGE \bf
Safety Embedded Adaptive Control Using Barrier States
}

\author{Maitham F. AL-Sunni \quad \quad Hassan Almubarak \quad \quad John M. Dolan
\thanks{M. F. AL-Sunni is with the Electrical and Computer Engineering Department, Carnegie Mellon University, Pittsburgh, PA,
USA. Email: \href{mailto:malsunni@andrew.cmu.edu}{\tt malsunni@andrew.cmu.edu}}%
\thanks{H. Almubarak is with the School of Electrical and Computer Engineering, Georgia Institute of Technology, Atlanta, GA, USA. Email: \href{mailto:halmubarak@gatech.edu}{\tt halmubarak@gatech.edu}}%
\thanks{J. M. Dolan is with the Robotics Institute, Carnegie Mellon University, Pittsburgh, PA, USA. Email: \href{mailto:jdolan@andrew.cmu.edu}{\tt jdolan@andrew.cmu.edu}}
  \thanks{This work has been accepted for publication in the
  \emph{Proceedings of the 2025 American Control Conference (ACC),
  Denver, Colorado, USA.}}
}

\usepackage{amsmath,amssymb,amsfonts}
\usepackage{algorithmic}
\usepackage{graphicx}
\usepackage{textcomp}
\usepackage{comment}

\usepackage{amsthm}
\usepackage[utf8]{inputenc}
\usepackage[T1]{fontenc}
\usepackage{enumerate}

\usepackage{soul}
\usepackage{xcolor}

\usepackage{hyperref}
\hypersetup{
    linkcolor=black,
    filecolor=black,      
    urlcolor=black,
}

\usepackage[style=ieee]{biblatex}
  
\AtBeginBibliography{\scriptsize}  
\newtheorem{theorem}{Theorem}

\newtheorem{lemma}{Lemma}
\theoremstyle{definition}
\newtheorem{definition}{Definition}

\newtheorem{remark}{Remark}

\usepackage{siunitx}
\DeclareSIUnit{\rad}{rad}

\usepackage{tikz}

\usepackage[style=ieee]{biblatex}

\AtBeginBibliography{\scriptsize}  
\begin{document}

\maketitle

\thispagestyle{empty}
\pagestyle{empty}






\begin{abstract}
In this work, we explore the application of barrier states (BaS) in the realm of safe nonlinear adaptive control. Our proposed framework derives barrier states for systems with parametric uncertainty, which are augmented into the uncertain dynamical model. We employ an adaptive nonlinear control strategy based on a control Lyapunov functions approach to design a stabilizing controller for the augmented system. The developed theory shows that the controller ensures safe control actions for the original system while meeting specified performance objectives. We validate the effectiveness of our approach through simulations on diverse systems, including a planar quadrotor subject to unknown drag forces and an adaptive cruise control system, for which we provide comparisons with existing methodologies.
\end{abstract}

\section{Introduction}
Safe control methods have increasingly gained attention in recent research due to their importance in ensuring system reliability. Many of these methods rely on the notion of set invariance and detailed system models to maintain safety. However, these model-based approaches often require highly accurate system representations, which can be difficult to achieve. Maintaining a precise model is challenging as physical systems vary over time due to factors like wear, environmental changes, and operational adjustments. Consequently, models may become outdated, leading to errors and uncertainties that can degrade the performance of traditional control strategies and jeopardize safety.

To address these challenges, this paper presents a framework that unifies adaptive control techniques with barrier states and safety‐embedded systems. In our approach, barrier states directly encode safety constraints into an augmented system, making the enforcement of safe operation an intrinsic part of the control design. An adaptation law compensates for uncertain parameters, ensuring that the system meets performance objectives while maintaining bounded barrier states, which in turn ensures safety. Specifically, we develop Barrier States Embedded Adaptive Control Lyapunov Functions (BaS‐aCLFs) to design the adaptation law, thereby stabilizing the augmented system under parameter uncertainty while enabling safe, multi‐objective control.

\begin{figure}
    \centering
    \includegraphics[trim=0 23 0 0, clip, width=1\linewidth]{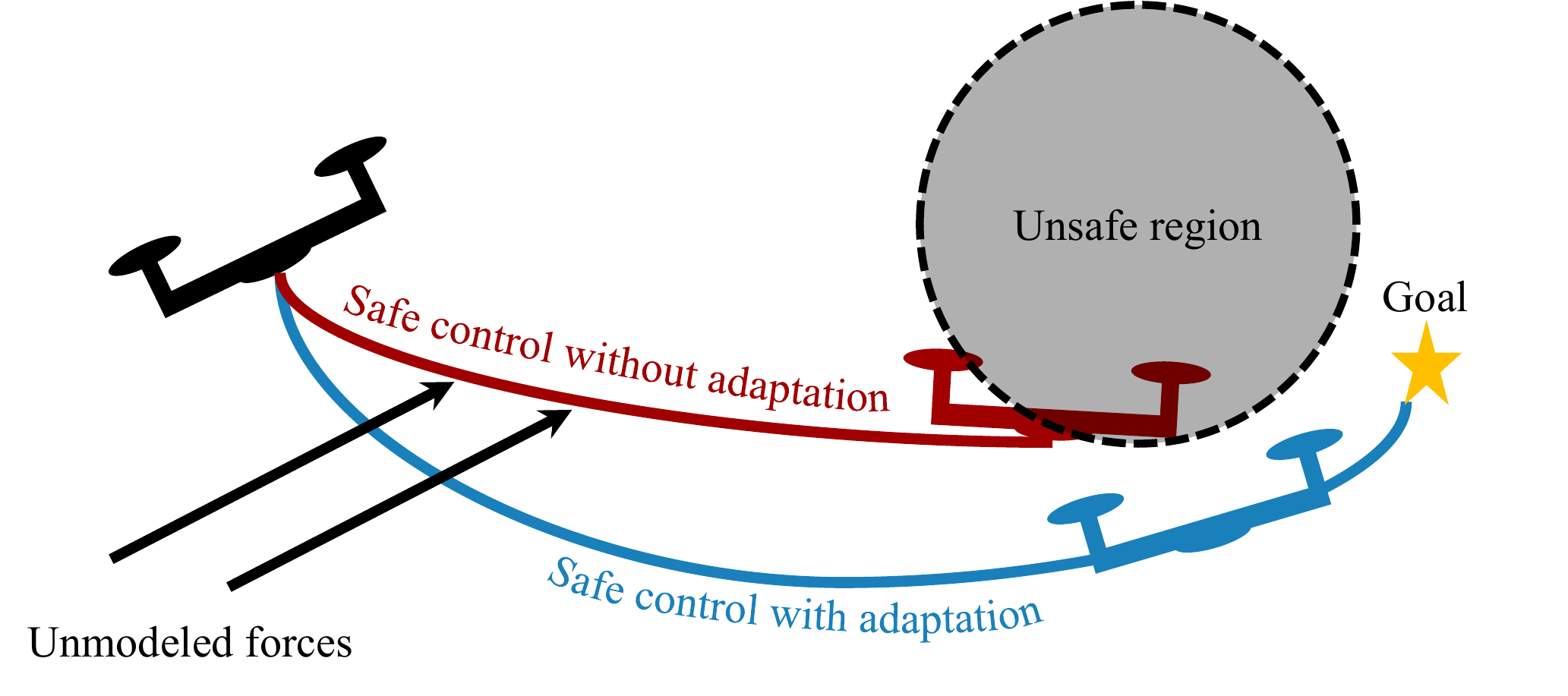}
    \caption{An illustration of our framework where the blue quadrotor can navigate safely even with unmodeled external forces affecting it. Plain safe control without adaptation (red quadrotor) cannot handle the scenario because it does not consider the presence of the unmodeled external forces.}
    \label{fig:quad_illistration}
\vspace{-7mm}
\end{figure}

\subsection{Related Work}
Control barrier functions (CBFs) \cite{wieland2007constructive,ames2016CBF-forSaferyCritControl,romdlony2016stabilization}, have proven to be a viable safe control approach. Safe controls through CBFs can be effectively achieved by solving a quadratic program that filters out unsafe nominal controls \cite{ames2016CBF-forSaferyCritControl}. CBF approaches continue to evolve, with a few open challenges remaining. These include synthesizing valid CBFs to be used with the CBF-QP filters \cite{clark2022semi, zhao2023safety,liu2024synthesis}, which is a tedious task for complex constraints; designing parameters that effectively balance performance and safety; CBFs for high relative degree constraints \cite{ames2016CBF-forSaferyCritControl}, which also make parameter selection more complicated; and handling control input limits \cite{liu2023safe}.

Barrier States (BaS) \cite{Almubarak2021SafetyEC} offer a promising alternative for safe control, addressing key limitations of Control Barrier Functions (CBFs). Unlike CBFs, BaS directly encode safety constraints, eliminating the need to synthesize a valid barrier function. By reformulating the system as a safety embedded system, BaS inherently regulate the barrier function’s dynamics, reducing performance-safety trade-offs. Moreover, BaS are agnostic to constraint relative degree, avoiding the complexities of high-relative-degree barriers \cite{almubarak2021safeddp,almubarak2023barrier}. While this approach increases model nonlinearity and state dimension, established nonlinear control techniques effectively mitigate these challenges \cite{almubarak2021safeddp,almubarak2023barrier}.

Despite advancements, both CBF and BaS methods require full knowledge of system dynamics, limiting their applicability to systems with unknown or uncertain parameters. This challenge has driven research on integrating barrier methods with adaptive control \cite{taylor2020adaptive,fan2020bayesian,lopez2020robust,konig2021safe,lopez2023unmatched,aoun2023l1}. Adaptive CBFs (aCBFs) \cite{taylor2020adaptive}, inspired by Adaptive Control Lyapunov Functions (aCLFs) \cite{krstic1995control}, introduced \textit{adaptive safety} to extend safe control to uncertain systems. However, aCBFs can be overly conservative, as they require the safety function’s rate of change to remain strictly positive, restricting trajectories from reaching the safe set’s boundary \cite{lopez2020robust}. To address this, robust adaptive CBFs (raCBFs) refine the safe set definition, allowing states to approach its boundary \cite{lopez2020robust}.

It is important to note that the open challenges previously discussed for CBFs—such as parameter tuning,  handling higher relative degrees, and control input limits—also apply to aCBFs and raCBFs, given their conceptual roots in the CBF framework. Additionally, to develop a distinct adaptation law based on the CBF approach, the safety function must be parameterized so that the adaptation gain influences the CBF condition. Moreover, raCBFs assume exact knowledge of the bounds of the unknown parameters, which may be a strong assumption for some applications. As for barrier states, the only effort in the literature towards safe adaptive control was in \cite{aoun2023l1}, where an $L_1$ filtering-based adaptive control method was developed, but was limited to linear systems with disturbances. 

\subsection{Contribution and Organization}
This paper seeks to extend BaS to systems with parametric uncertainty, providing a framework for safe adaptive control that enables multi-objective control. To the best of our knowledge, this is the first work to adapt BaS to handle systems with unknown parameters within an adaptive control framework, presenting a novel approach that addresses the discussed limitations while maintaining safety. Using barrier states, we derive a parameterized safety embedded system with characteristics similar to those of the original system. This allows us to use legacy adaptive nonlinear control methods to simultaneously address safety and performance objectives despite the uncertainty in the system's model.

The paper is organized as follows. \autoref{sec: prelim bas and multi obj control} gives the preliminaries of barrier states. \autoref{sec:problem_statement} presents our problem statement. \autoref{sec:main} provides the main result of the paper in which we show how barrier states integrate the concepts of safety and adaptive control into a unified framework that facilitates multi-objective adaptive control. Finally, in \autoref{sec:simulations}, we verify our framework via numerical simulations with three different systems: an inverted pendulum system, an adaptive cruise control (ACC) system, and a planar quadrotor system.

\section{Preliminaries: Barrier States for Multi-objective Control} \label{sec: prelim bas and multi obj control}
Embedded barrier states originated as a model-based method that enables multi-objective control for safety-critical systems. The concept suggests constructing a \textit{safety state} of the system that can be controlled along with the other states. Essentially, barrier states (BaS) establish barriers within the state space of a new system, directing the development of a multi-objective feedback control law towards the realm of safe controls. BaS are augmented into the model of the safety-critical dynamical system to produce an augmented system, referred to as the safety embedded system, that achieves safety if and only if its trajectories are bounded. For instance, in stabilization scenarios, safety is assured when the equilibrium point of interest in the augmented system is asymptotically stabilized. Consequently, with the integration of safety and stability, designing a controller that stabilizes the augmented model also ensures a safe control solution for the original system. This integration means that the safety characteristics are inherent in the system's stability, as the feedback controller depends on both the system’s states and the barrier states—hence the term \textit{safety embedded control}.

Consider the system
\begin{align} \label{eq:bas_prelim_system}
    \dot{x} = f(x,u),
\end{align}
where $x \in \mathcal{X} \subset \mathbb{R}^n$ is the state, $u \in \mathcal{U} \subset \mathbb{R}^m$ is the control input, $f:\mathcal{X} \times \mathcal{U} \to \mathcal{X}$ is continuously differentiable, and without loss of generality, we assume that $f(0,0) = 0$. 
\begin{definition} \label{def:bf}
    The function $\pmb{B} : \mathbb{R} \to \mathbb{R}$ is a barrier function if it is smooth on $(0, \infty)$ and $\pmb{B}(\eta) \xrightarrow[]{\eta \rightarrow 0} \infty$. Examples for barrier functions are: the inverse barrier function $\pmb{B}(\eta) = \frac{1}{\eta}$ and the  logarithmic barrier function $\pmb{B}(\eta) = -\text{log}\left(\frac{\eta}{1+\eta} \right)$.
\end{definition}

The system \eqref{eq:bas_prelim_system} is subject to the safe set $\mathcal{S}$ defined by the safety function $h(x)$. Namely, the safety set is defined as ${\mathcal{S}} \triangleq \left\{x \in \mathcal{X} | h(x) > 0  \right\}$ with $\partial{\mathcal{S}} \triangleq \left\{x \in \mathcal{X} | h(x) = 0 \right\}$ being the boundary set of $\mathcal{S}$. We define $\beta(x)=\pmb{B}(h(x))$, where $\pmb{B}$ is any valid barrier function. Note that $\beta(x) \to \infty$ if and only if $h(x) \to 0$ (approaching unsafe regions) \cite{Almubarak2021SafetyEC}. The barrier function's time derivative\footnote{For notational convenience we use $\pmb{B}'(\eta)$ to denote $\frac{d\pmb{B}}{d\eta}(\eta)$.} is given by $\dot{\beta}(x) = \pmb{B}'\big(\pmb{B}^{-1}(\beta) \big) L_{f(x,u)}h(x)$. Then, the \textit{barrier state}, denoted as $z$, is defined by the state equation
\begin{equation} \label{eq: deterministic BaS z}
  \dot{z} = f^z := \pmb{B}'\big(\pmb{B}^{-1}(z+\beta^{0}) \big) L_{f(x,u)}h(x) - \gamma \big(z+\beta^{0} - \beta(x) \big),
\end{equation}
where $\gamma \in \mathbb{R}^+_0$, $\beta^{0} =\beta(0)$ and the shift by $\beta^{0}$ ensures that $z=0$ is an equilibrium state. It is important to note that when the initial condition $z(0)=\beta\big(x(0)\big) - \beta^0$, we have the solution $z=\beta(x) - \beta^0$ (the reader may refer to \cite{almubarak2023barrier} for detailed derivations of the barrier states and its theory). It must be noted that the BaS $z$ is bounded if and only if the barrier function $\beta(x)$ is bounded \cite[Proposition~2]{Almubarak2021SafetyEC}, which implies that ensuring the boundedness of the BaS guarantees safety.

For a $\mathcal{Q}$ number of constraints, one can design a single barrier state by constructing an aggregated safety function $h$ such that $\frac{1}{h} = \sum_{i=1}^{\mathcal{Q}} \frac{1}{h_i}$ or use
multiple barrier states \cite{almubarak2023barrier,Almubarak2021SafetyEC,almubarak2021safeddp}. For generality, let $z=[z_1, \dots, z_q]^{\top} \in \mathcal{Z} \subseteq \mathbb{R}^q$ where $q$ is the number of barrier states. 

The barrier state equation \eqref{eq: deterministic BaS z} is augmented to the model of the system resulting in the \textit{safety embedded system}
\begin{equation}
\label{eq:safety_augmented_prelim}
\dot{\bar{x}}= \bar{f}(\bar{x}, u),
\end{equation}
where $\bar{x}=\begin{bmatrix} x \\ z \end{bmatrix} \in \bar{\mathcal{X}} \subset \mathcal{X} \times \mathcal{Z}$, $\bar{f}=\begin{bmatrix} f \\ f ^z\end{bmatrix} \in \mathcal{\bar{X}} \times \mathcal{U} \to \mathcal{\bar{X}}$ with $\bar{f}(0)=0$. The safety embedded system is continuously differentiable and its origin is stabilizable \cite{almubarak2023barrier,Almubarak2021SafetyEC}. Therefore, the safety constraint is \textit{embedded} in the closed-loop system's dynamics, and stabilizing the BaS implies enforcing safety for the safety-critical system \eqref{eq:bas_prelim_system}. Hence, stabilizing the origin of the safety embedded system \eqref{eq:safety_augmented_prelim} entails safely stabilizing the origin of the safety-critical system \eqref{eq:bas_prelim_system}.

\begin{theorem}[\cite{almubarak2023barrier}]
The original control system \eqref{eq:bas_prelim_system} is safely stabilizable at the origin if and only if the safety-embedded control system \eqref{eq:safety_augmented_prelim} is stabilizable at the origin\footnote{Note that safe stabilizability implies that the origin is in the safe set.}.
\label{th:if_org_then_se}
\end{theorem}

This approach is versatile and can be readily adopted in conjunction with established control methods, enabling its application to general nonlinear systems and constraints in safety-critical multi-objective control and planning. \cite{Almubarak2021SafetyEC,almubarak2021safeddp,song2023safety,aoun2023l1,cho2023model, oshin2024differentiable}. Another major advantage of this technique is that the results of \autoref{th:if_org_then_se} are generally agnostic to the relative degree of the barrier function, which is a critical part in the design of CBFs that can complicate the safe control design problem. Furthermore, the rate of change of the barrier is governed by the control design\footnote{In CBF approaches \cite{ames2019control}, the rate of change is decided by the choice of the class $\mathcal{K}_\infty$ (which is usually called $\alpha$) which can make the performance very conservative in many cases.}, which aims to balance performance and safety objectives, rather than being manually adjusted. In this work, we adopt this concept within an adaptive control framework to ensure safe control. 

\section{Problem Statement} \label{sec:problem_statement}

Consider the dynamical system
\begin{align} \label{eq:org_system}
    \dot{x} =  f(x) + g(x)u + F(x) \theta, 
\end{align}
where $x \in \mathcal{X} \subset \mathbb{R}^n$ is the state, $u \in \mathcal{U} \subset \mathbb{R}^m$ is the control input, $\theta \in \mathbb{R}^p$ is an unknown parameter which can take any value, with $f:\mathcal{X} \to \mathcal{X}$, $g:\mathcal{X} \to \mathbb{R}^{n \times m}$, and $F:\mathcal{X} \to \mathbb{R}^{n \times p}$ being continuously differentiable. It is assumed that the origin is an equilibrium point of the system with $F(0) = 0$. It is desired to ensure that the system stays in the open safe set $\mathcal{S}$, which is a superlevel set defined by the continuously differentiable function $h(x)>0$. It must be noted that the uncertainty of the model means that safety guarantees of model-based safe control methods may be jeopardized and further considerations are needed. Later we will see how such uncertainty in the model affects the embedded BaS method.

In this paper, we aim to extend the concept of multi-objective safe control via barrier states to systems with unknown or uncertain parameters and develop a safe nonlinear adaptive control framework in which safety and other desired performance objectives such as stabilization are achieved. Our goal is to design an adaptive controller of the form
\begin{align} \label{eq:adaptive_controller and adaptation} 
    u = K(x,\hat{\theta}), \quad 
    \dot{\hat{\theta}} = \Gamma \tau(x,\hat{\theta}), 
\end{align}
where $\hat{\theta} \in \mathbb{R}^p$ is an estimate of the parameter $\theta$, $\Gamma \in \mathbb{R}^{p \times p}$ is a positive definite gain, $K: \mathcal{X} \times \mathbb{R}^p \to \mathbb{R}^m$, and $\tau : \mathcal{X} \times \mathbb{R}^p \to \mathbb{R}^p$, such that we achieve safe adaptive control.

\begin{definition}[Safe Adaptive Control]
    The controller \eqref{eq:adaptive_controller and adaptation} is a safe adaptive controller for the system \eqref{eq:org_system}, if the solution of the closed loop system $\dot{x} =  f(x) + g(x)K(x,\hat{\theta}) + F(x) \theta$ with $\dot{\hat{\theta}} = \Gamma \tau(x,\hat{\theta})$ satisfies $x(t) \in {\mathcal{S}}, \forall t \in [0,T]$ for any $\theta$ given that $x(0) \in \mathcal{S}$, where $T$ is the total running time of the system, implying forward invariance of the safe set $\mathcal{S}$.
    \label{def:safe_adaptive_control}
\end{definition}

It is worth noting that previous work in \cite{taylor2020adaptive,lopez2020robust,lopez2023unmatched} proposed using a parameterized safety function, mimicking the aCLF stabilization formulation \cite{krstic1995control}. Specifically, the function $h$ was parameterized by $\theta$, which is not necessarily meaningful in the case of safe control unless $h$ contains uncertain parameters, which was not the case in their development nor was it used in the numerical simulations. Due to the nature of CBF as a control-dependent constraint, CBF-based approaches proposed control frameworks that depend on two adaptation laws. One adaptation law is dedicated to the parameter estimate used by the adaptive CBF methods, while the second adaptation law is dedicated to the parameter estimate used to achieve the control task, which typically calls for a trade-off and manual tuning of the class $\mathcal{K}$ functions. On the other hand, the proposition of embedded barrier states aims towards \textit{embedding} the safety objective within the control objectives. This yields the various advantages discussed earlier. The proposed framework in this paper integrates the concepts of safety via barrier methods and adaptive control through one adaptation law. The benefits of this approach are discussed in the next section. 

\section{Safety Embedded Adaptive Control} \label{sec:main}
In this section, we explore the application of BaS in adaptive control, by embedding the dynamics of the uncertain BaS into the original uncertain system, we create a safety-embedded system that retains the structure of the original formulation \eqref{eq:org_system}. This enables us to tackle the adaptive control problem, where the objective is to devise an adaptation law that stabilizes the states of the embedded system despite the uncertainty, thereby ensuring the safety of the original system. Remarkably, this approach facilitates a direct unification of safety and control in the face of parametric uncertainty.

Building on the barrier states formulation reviewed in the preliminaries, \autoref{sec: prelim bas and multi obj control}, the BaS equation for the system with parametric uncertainty \eqref{eq:org_system} is derived as follows:
the dynamics of $\beta(x) = \pmb{B}(h(x))$ is given by 
\begin{align}    
\dot{\beta}(x) &= \pmb{B}'\big(\pmb{B}^{-1}(\beta) \big) L_{f(x)}h(x) + \pmb{B}'\big(\pmb{B}^{-1}(\beta) \big) L_{g(x)}h(x)u \nonumber \\ &+ \pmb{B}'\big(\pmb{B}^{-1}(\beta) \big) L_{F(x)}h(x)\theta.
\end{align}
Then, the BaS equation is given by
\begin{align} \label{eq: bas with theta star}
    \dot{z} &= \underbrace{\pmb{B}'\left( \pmb{B}^{-1}(z + \beta^{0}) \right) L_{f(x)}h(x) - \gamma \big(z + \beta^{0} - \beta(x)  \big)}_{f_z} \nonumber \\
    &+ \underbrace{\pmb{B}'\left( \pmb{B}^{-1}(z + \beta^{0})\right) L_{g(x)}h(x)}_{g_z} u \nonumber \\ 
    &+ \underbrace{\pmb{B}' \left( \pmb{B}^{-1}(z + \beta^{0})\right) L_{F(x)}h(x)}_{F_z} \theta,
\end{align}
with the initial condition $z(0) = \beta(x(0))-\beta^0$,  where $\beta^0 = \beta(0)$. The BaS equation \eqref{eq: bas with theta star} can be put in the form
\begin{align} \label{eq: bas short form theta star}
    \dot{z} &= f_z + g_z u + F_z \theta.
\end{align}
\begin{lemma} \label{lemma: safety if bas is bounded}
For the original system \eqref{eq:org_system}, given that $x(0) \in \mathcal{S}$, an adaptive controller $u = K({x},\hat{\theta})$ is safe, i.e. $h(x) > 0$ is satisfied and the safe set $\mathcal{S}$ is forward invariant, if and only if  $z(x(t)) < \infty \ \forall t \in [0,T]$.
\end{lemma} 
\begin{proof}
    $\Rightarrow$ Let there exists a control law $u = K({x},\hat{\theta})$ such that $\mathcal{S}$ is forward invariant for the closed-loop system $\dot{{x}}= {f}({x}) + {g}({x})K({x},\hat{\theta}) + {F}({x})\theta$ with the solution $x(t) \in \mathcal{S} \ \forall t \in [0,T]$. By definition of the safe set, we have $h(x) > 0 \ \forall t \in [0,T]$ Hence, by \autoref{def:bf} and the properties of the barrier function $\beta$, $\beta(x(t)) < \infty$. Hence, by \cite[Proposition~2]{Almubarak2021SafetyEC}, $z(x(t)) < \infty\ \ \forall t \in [0,T]$. \\ 
    $\Leftarrow$ Let $h\left(x(0)\right)>0 \Rightarrow \beta\left(x(0)\right) < \infty$ and suppose that $z < \infty \ \forall t \in [0,T]$ under some continuous controller $u = K({x},\hat{\theta})$. By \cite[Proposition~2]{Almubarak2021SafetyEC}, we have $\beta\left(x(t)\right) < \infty \ \forall t \in [0,T]$. Then, by \autoref{def:bf} and the definition of $\beta$, $h\left(x(t)\right)>0 \ \forall t \in [0,T]$. As a consequence, $\mathcal{S}$ is forward invariant with respect to the closed loop system $\dot{{x}}= {f}({x}) + {g}({x})K({x},\hat{\theta}) + {F}({x})\theta$ and hence $u$ is safe.
\end{proof}

As ensuring boundedness of the BaS $z$ guarantees safety of the system, we aim to design an adaptive controller that renders the BaS in \eqref{eq: bas short form theta star} bounded regardless of the uncertain parameter $\theta$. Consequently, the idea is to develop an adaptive controller of the form \eqref{eq:adaptive_controller and adaptation} that stabilizes the states of the system along with the barrier states. More specifically, given that \eqref{eq: bas short form theta star} has an identical form to \eqref{eq:org_system}, we work with the safety embedded system with $q$ barrier states
\begin{equation}
\label{eq:safety_augmented_sys_ad_saf}
\dot{\bar{x}}= \bar{f}(\bar{x}) + \bar{g}(\bar{x})u + \bar{F}(\bar{x})\theta,
\end{equation}
where $\bar{x}=\begin{bmatrix} x \\ z \end{bmatrix} \in \bar{\mathcal{X}} \subset \mathcal{X} \times \mathcal{Z} $, $\bar{f}=\begin{bmatrix} f \\ f_z\end{bmatrix}: \bar{\mathcal{X}} \to \bar{\mathcal{X}}$, $\bar{g}=\begin{bmatrix} g \\ g_z\end{bmatrix}: \bar{\mathcal{X}} \to \mathbb{R}^{(n+q)\times m}$, and $\bar{F}=\begin{bmatrix} F \\ F_z\end{bmatrix}: \bar{\mathcal{X}} \to \mathbb{R}^{(n+q)\times p}$.

Now, the objective is to design a controller that adaptively stabilizes the safety embedded system \eqref{eq:safety_augmented_sys_ad_saf} regardless of $\theta$.


The closed-loop safety embedded system under the adaptive controller \eqref{eq:adaptive_controller and adaptation} can be written as
\begin{align}
    \begin{bmatrix}
    \dot{\bar{x}} \\ 
    \dot{\hat{\theta}}
    \end{bmatrix}
    =
    \begin{bmatrix}
    \bar{f}(\bar{x}) + \bar{g}(\bar{x})K(\bar{x},\hat{\theta})+ \bar{F}(\bar{x}) \theta \\
    \Gamma \tau(\bar{x},\hat{\theta})
    \end{bmatrix}.
    \label{eq:closed_loop_ac}
\end{align}

The control law $u=K(\bar{x},\hat{\theta})$ needs to be determined to deal with the parametric uncertainty in \eqref{eq:safety_augmented_sys_ad_saf}. In the adaptive control literature, several techniques have been proposed. This includes adaptive control Lyapunov functions (aCLFs) \cite{krstic1995control}, universal CLFs \cite{lopez2021universal}, differentiable
robust model predictive control \cite{oshin2024differentiable}, and tube model predictive control \cite{mayne2011tube}. 

In this work, we get inspiration from (aCLF) to formulate Barrier States Embedded Adaptive Control Lyapunov Functions (BaS-aCLFs).

\begin{definition} \label{def:bas_aclf}
Let there exist a positive definite symmetric matrix $\Gamma$ such that for each $\theta \in \mathbb{R}^p$, the system 
    \begin{equation} \label{eq: modified system}
        \dot{\bar{x}} = \bar{f}(\bar{x}) + \bar{g}(\bar{x})u + \bar{F}(\bar{x}) \vartheta(\bar{x}, \theta), 
    \end{equation}
where $\vartheta(\bar{x}, \theta) = \left(\theta + \Gamma \left(\frac{\partial V_{\text{Ba}}}{\partial \theta}\right)^{\top} \right)$ possesses the Control Lyapunov Function (CLF) $V_{\text{Ba}}(\bar{x}, \theta)$. That is,
\begin{align*} 
\inf_{u \in \mathcal{U}}  \frac{\partial V_{\text{Ba}}}{\partial \bar{x}} \Bigg[\bar{f}(\bar{x}) + \bar{g}(\bar{x})u + \bar{F}(\bar{x}) \vartheta(\bar{x}, \theta)  \Bigg] \leq - \alpha_1(||\bar{x}||,\theta)
\end{align*}
and $\alpha_2(||\bar{x}||,\theta) \leq V_{\text{Ba}}(\bar{x}, \theta) \leq \alpha_3(||\bar{x}||,\theta)$, where $\alpha_1$, $\alpha_2$, and $\alpha_3$ are class $\mathcal{K}_\infty$ functions. Then, $V_{\text{Ba}}$ is called a BaS-aCLF for \eqref{eq:safety_augmented_sys_ad_saf}. 
\end{definition}

\begin{lemma}
    \label{lemma:if_aclf_then_safe}
    There exists an adaptive controller  $u = K(\bar{x}, \hat{\theta})$ that adaptively stabilizes the origin of \eqref{eq: modified system} for any $\theta \in \mathbb{R}^p$ if there exists a BaS-aCLF for \eqref{eq:safety_augmented_sys_ad_saf}. Furthermore, $u = K(\bar{x}, \hat{\theta})$ renders the origin of the safety embedded system \eqref{eq:safety_augmented_sys_ad_saf}  adaptively asymptotically stable. 
\end{lemma}
\begin{proof}  
    Let $V_{\text{Ba}}$ be a valid BaS-aCLF for the safety embedded system \eqref{eq:safety_augmented_sys_ad_saf}, and consider the candidate composite Lyapunov function for the closed-loop system \eqref{eq:closed_loop_ac}
    \begin{align}
        V(\bar{x},\hat{\theta}) = V_{\text{Ba}}(\bar{x},\hat{\theta}) + \frac{1}{2} \tilde{\theta}^\top \Gamma^{-1} \tilde{\theta},
    \end{align}
    where $\tilde{\theta} = \theta - \hat{\theta}$. Taking its time derivative gives
    {\small \begin{align}
        \dot{V}
        &= \frac{\partial V_{\text{Ba}}}{\partial \bar{x}} \left[\bar{f}(\bar{x}) + \bar{g}(\bar{x})K(\bar{x},\hat{\theta}) + \bar{F}(\bar{x}) \theta  \right]  \\
        & + \frac{\partial V_{\text{Ba}}}{\partial \hat{\theta}}\Gamma\tau(\bar{x},\hat{\theta}) \nonumber - \tilde{\theta}^\top \tau(\bar{x},\hat{\theta}) \nonumber \\
        &= \frac{\partial V_{\text{Ba}}}{\partial \bar{x}} \left[\bar{f}(\bar{x}) + \bar{g}(\bar{x})K(\bar{x},\hat{\theta}) + \bar{F}(\bar{x}) \vartheta(\bar{x},\hat{\theta})  \right] - \tilde{\theta}^\top \tau(\bar{x},\hat{\theta}) \nonumber \\ 
        &+ \frac{\partial V_{\text{Ba}}}{\partial \hat{\theta}}\Gamma\tau(\bar{x},\hat{\theta}) + \frac{\partial V_{\text{Ba}}}{\partial \bar{x}} \bar{F} \tilde{\theta} - \frac{\partial V_{\text{Ba}}}{\partial \hat{\theta}} \Gamma \left( \frac{\partial V_{\text{Ba}}}{\partial \bar{x}} \bar{F} \right)^\top. \nonumber
        \end{align}}        
        By \autoref{def:bas_aclf},
        \begin{align}
                \dot{V} &\leq -\alpha_1(||\bar{x}||,\hat{\theta}) + \tilde{\theta}^\top \psi(\bar{x},\hat{\theta}) - \left(\frac{\partial V_{\text{Ba}}}{\partial \hat{\theta}}(\bar{x},\hat{\theta}) \right) \Gamma \psi(\bar{x},\hat{\theta}), \nonumber
    \end{align}
    where $\psi(\bar{x},\hat{\theta}) = \left(\frac{\partial V_{\text{Ba}}}{\partial \bar{x}}(\bar{x},\hat{\theta})\bar{F}(\bar{x}) \right)^\top - \tau(\bar{x},\hat{\theta})$. For the adaptive controller, let the adaptation function be
    \begin{align}
    \tau(\bar{x},\hat{\theta}) = \left(\frac{\partial V_{\text{Ba}}}{\partial \bar{x}}(\bar{x},\hat{\theta})\bar{F}(\bar{x}) \right)^\top,
    \end{align}
    which results in $\dot{V} \leq -\alpha_1(||\bar{x}||,\hat{\theta})$. Hence, under the adaptive controller $u=K(\bar{x}, \hat{\theta})$ with the adaptation law $\dot{\hat{\theta}}= \Gamma \left(\frac{\partial V_{\text{Ba}}}{\partial \bar{x}}(\bar{x},\hat{\theta})\bar{F}(\bar{x}) \right)^\top$, the origin of the closed-loop system in \eqref{eq:closed_loop_ac}, $\bar{x}=0,\hat{\theta} = \theta$ is globally stable and $\bar{x}(t)\to 0$ by the invariance principle of LaSalle \cite{khalil2002nonlinear}. Furthermore, the safety embedded system \eqref{eq:safety_augmented_sys_ad_saf} is adaptively asymptotically stable for any $\theta \in \mathbb{R}^p$. 
\end{proof}

\begin{remark}
Although $(\bar{x},\hat{\theta})=(0,\theta)$ is a Lyapunov stable equilibrium of the closed-loop system,
it needs not be the unique attractor. Without further restrictive assumptions, $\hat{\theta}(t)$ may converge to 
some value different from $\theta$ while still driving $\bar{x}(t)\to0$.
\end{remark}

\begin{theorem}
    Let $x(0) \in \mathcal{S}$ with $V_{\text{Ba}}$ being a BaS-aCLF for the safety embedded system \eqref{eq:safety_augmented_sys_ad_saf} under the controller $u = K(\bar{x},\hat{\theta})$, $\dot{\hat{\theta}} = \Gamma \tau(\bar{x},\hat{\theta})$ with 
    \begin{align}
         \tau(\bar{x},\hat{\theta}) = \left(\frac{\partial V_{\text{Ba}}}{\partial \bar{x}}(\bar{x},\hat{\theta})\bar{F}(\bar{x}) \right)^\top.
    \end{align}
Then, the original system \eqref{eq:org_system} is adaptively asymptotically stable and safe for any $\theta \in \mathbb{R}^p$. 
\end{theorem}
\begin{proof}
    By \autoref{lemma:if_aclf_then_safe}, the BaS-aCLF ensures that the origin of the safety embedded system is adaptively asymptotically stable ($\bar{x}\rightarrow 0$) for any $\theta \in \mathbb{R}^p$. Hence, the barrier state in \eqref{eq: bas short form theta star} remains bounded. Furthermore, by \autoref{lemma: safety if bas is bounded}, boundedness of $z(t)$ implies that the controller $u = K(\bar{x}, \hat{\theta})$ is a safe adaptive controller ensuring that $x(t) \in \mathcal{S}$ for all $t \in [0, T]$.
\end{proof}
Given a BaS-aCLF, one could implement an optimization-based controller by solving the following quadratic program:
\begin{align}
    {K}(\bar{x}, \hat{\theta}) &= \underset{{u} \in \mathcal{U}}{\arg\min} \quad \frac{1}{2} \|{u}\|^2 \\ \hspace{-5mm}
&\text{s.t.} \quad \frac{\partial V_{\text{Ba}}}{\partial \bar{x}}(\bar{x}, \hat{\theta}) 
\left(\bar{f}(\bar{x}) + \bar{g}(\bar{x}) {u} + \bar{F}(\bar{x}) \vartheta(\bar{x}, \hat{\theta})  \right) \nonumber \\
&\leq - \alpha_1(\|\bar{x}\|, \hat{\theta}). \nonumber
\end{align}
Via KKT conditions \cite{boyd2004convex}, the closed-loop solution of the QP is given by
\begin{align}
        {K}(\bar{x}, \hat{\theta}) = \left\{\begin{array}{lr}
        -\frac{\mu(\bar{x},\hat{\theta})\nu(\bar{x},\hat{\theta})}{\nu^\top(\bar{x},\hat{\theta}) \nu(\bar{x},\hat{\theta})}, &\text{ if } \mu(\bar{x},\hat{\theta}) > 0, \\
        0,  &\text{ if } \mu(\bar{x},\hat{\theta}) \leq 0,
        \end{array}\right.
\end{align}
where $\mu(\bar{x},\hat{\theta}) \triangleq
    \frac{\partial V_{\text{Ba}}}{\partial \bar{x}}(\bar{x}, \hat{\theta}) \left(\bar{f}(\bar{x})  + \bar{F}(\bar{x})\vartheta(\bar{x}, \hat{\theta}) \right) + \alpha_1(\|\bar{x}\|, \hat{\theta})$ and $\nu(\bar{x},\hat{\theta}) \triangleq
    \frac{\partial V_{\text{Ba}}}{\partial \bar{x}}(\bar{x}, \hat{\theta}) \bar{g}(\bar{x})$. This unified approach eliminates the need for multiple adaptation laws and offers a systematic method to integrate safety and performance objectives without the excessive conservatism of traditional approaches. Theoretical guarantees of safety and stability were provided through Lyapunov-based analysis, and the framework was shown to be applicable to general nonlinear systems.

\section{Numerical Simulations} \label{sec:simulations}
We validate our method on three systems: (1) a planar quadrotor with unknown drag forces, (2) an inverted pendulum with unknown viscous damping and gravity, and (3) adaptive cruise control with uncertain rolling friction. Our approach is compared to \textit{vanilla BaS} (which ignores uncertain parameters) and \textit{raCBF} \cite{lopez2020robust}, where applicable (as raCBFs cannot handle high-relative-degree constraints). 

In all experiments, we initialize parameter estimates as $\hat{\theta}(0) = \theta/10 $, use an adaptation gain $\Gamma = I$ (identity matrix), set $\gamma = 1$ for barrier states, and choose the barrier function $\pmb{B}(\eta) = 1/\eta$. The function $V_{\text{Ba}}$ is designed using a tuned Lyapunov function based on an LQR for the nominal safety-embedded system, i.e. without the uncertain parameters.

\subsection{Planar quadrotor}
Consider the quadrotor dynamics given by
\begin{align}
    \begin{bmatrix}
        \Ddot{x} \\ 
        \Ddot{y} \\
        \Ddot{\psi}
    \end{bmatrix}  = 
    \begin{bmatrix}
        \frac{1}{m} (u_1 + u_2) \sin{(\psi)} + d_x \dot{x} \\ 
        \frac{1}{m} (u_1 + u_2) \cos{(\psi)} - g + d_y \dot{y} \\
        \frac{l}{2J} (u_2 - u_1)   \\
    \end{bmatrix},
\end{align}
where ($x,y$) is the location, $\psi$ is the orientation, $m$ is the mass of the quadrotor, $l$ is the distance from the center to the propellers of the quadrotor, $J$ is the moment of inertia, $g$ is the acceleration due to gravity, $d_x$ and $d_y$ are the unknown drag coefficient in the $x$ and $y$ directions, and $u_1$ and $u_2$ denote the right and left thrust forces respectively. The objective is to safely thrust the quadrotor to hover from a random initial condition where the quadrotor might flip. Specifically, we want to safely stabilize the quadrotor at $x=\SI[per-mode = symbol]{0}{\meter}, y=\SI[per-mode = symbol]{1}{\meter}, \psi = \SI[per-mode = symbol]{0}{\rad}$, where the safe set is given by $y \geq \sqrt{0.5} \SI[per-mode = symbol]{}{\ \meter}$. It is important to note that with this system and constraint, we have a high relative degree that prevents the direct use of aCBFs and raCBFs.

For this example, we use the following parameters: $m=\SI[per-mode = symbol]{1}{\kilogram}$, $l=\SI[per-mode = symbol]{0.3}{\meter}$, $g = \SI[per-mode = symbol]{9.81}{\meter\per\square\second}$, $J = 0.2 \cdot m \cdot l^2$, $d_x = \SI[per-mode = symbol]{1}{\second^{-1}}$, and $d_y = \SI[per-mode = symbol]{1}{\second^{-1}}$. \autoref{fig:quadrotor_sim} shows that the proposed controller successfully adapts to overcome the uncertainty in the model while the barrier states embedded controller without adaptation fails to ensure the safety of the quadrotor.

\begin{figure}[ht]
\vspace{-4mm}
    \centering
    \includegraphics[width=\linewidth]{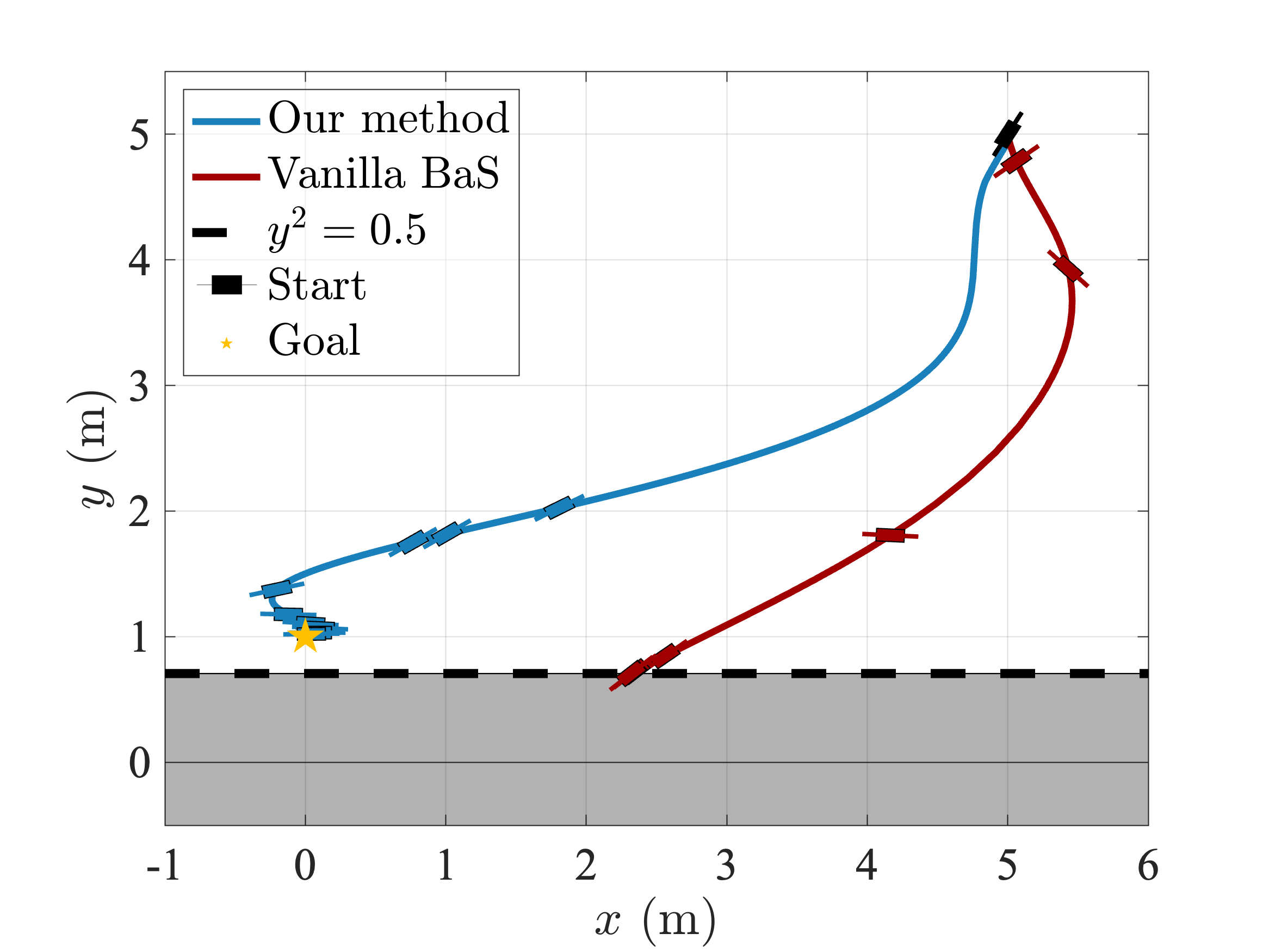}
    \caption{Our method successfully stabilizes the quadrotor at $x=\SI[per-mode = symbol]{0}{\meter}, y=\SI[per-mode = symbol]{1}{\meter}, \psi = \SI[per-mode = symbol]{0}{\rad}$ while BaS without adaptation cannot handle the task because of the uncertain $d_x$ and $d_y$.}
    \label{fig:quadrotor_sim}
\end{figure}
\vspace{-2mm}
\subsection{Inverted pendulum}
The dynamics are given by:
\begin{align}
    \begin{bmatrix}
        \dot{q} \\ 
        \Ddot{q}
    \end{bmatrix}  = 
    \begin{bmatrix}
        \dot{q} \\ 
        \frac{g}{l} \sin{(q)} - \frac{b}{ml^2} \dot{q}
    \end{bmatrix} +
    \begin{bmatrix}
        0 \\ 
        \frac{1}{m}l^2
    \end{bmatrix} u,
\end{align}
where $q$ and $\dot{q}$ are the pendulum’s angular position and velocity, respectively, and $m$, $l$, and $u$ denote its mass, length, and applied torque. The parameters $g$ (gravity) and $b$ (viscous damping) are unknown. The safety constraint is $q \leq \sqrt{\frac{\pi}{4}} \SI[per-mode = symbol]{}{\ \rad}$. As in the quadrotor example, the high relative degree prevents direct application of aCBFs and raCBFs.

For our example, we use the following parameters: $m=\SI[per-mode = symbol]{1}{\kilogram}$, $l=\SI[per-mode = symbol]{2}{\meter}$, $g = \SI[per-mode = symbol]{9.81}{\meter\per\square\second}$, and $b = \SI[per-mode = symbol]{1}{\kilogram\square\meter\per\second}$. In \autoref{fig:pendulum_sim}, we show an example of a stabilization task, the pendulum starts from $q = \frac{\pi}{5} \SI[per-mode = symbol]{}{\ \rad}$ and the goal is to stabilize around $q = \SI[per-mode = symbol]{0}{\rad}$. Again, we see how our method is able to handle the situation while barrier states alone, without the consideration of the unknowns, fail to avoid safety violations.
\begin{figure}[ht]
\vspace{-2mm}
    \centering
    \includegraphics[width=\linewidth]{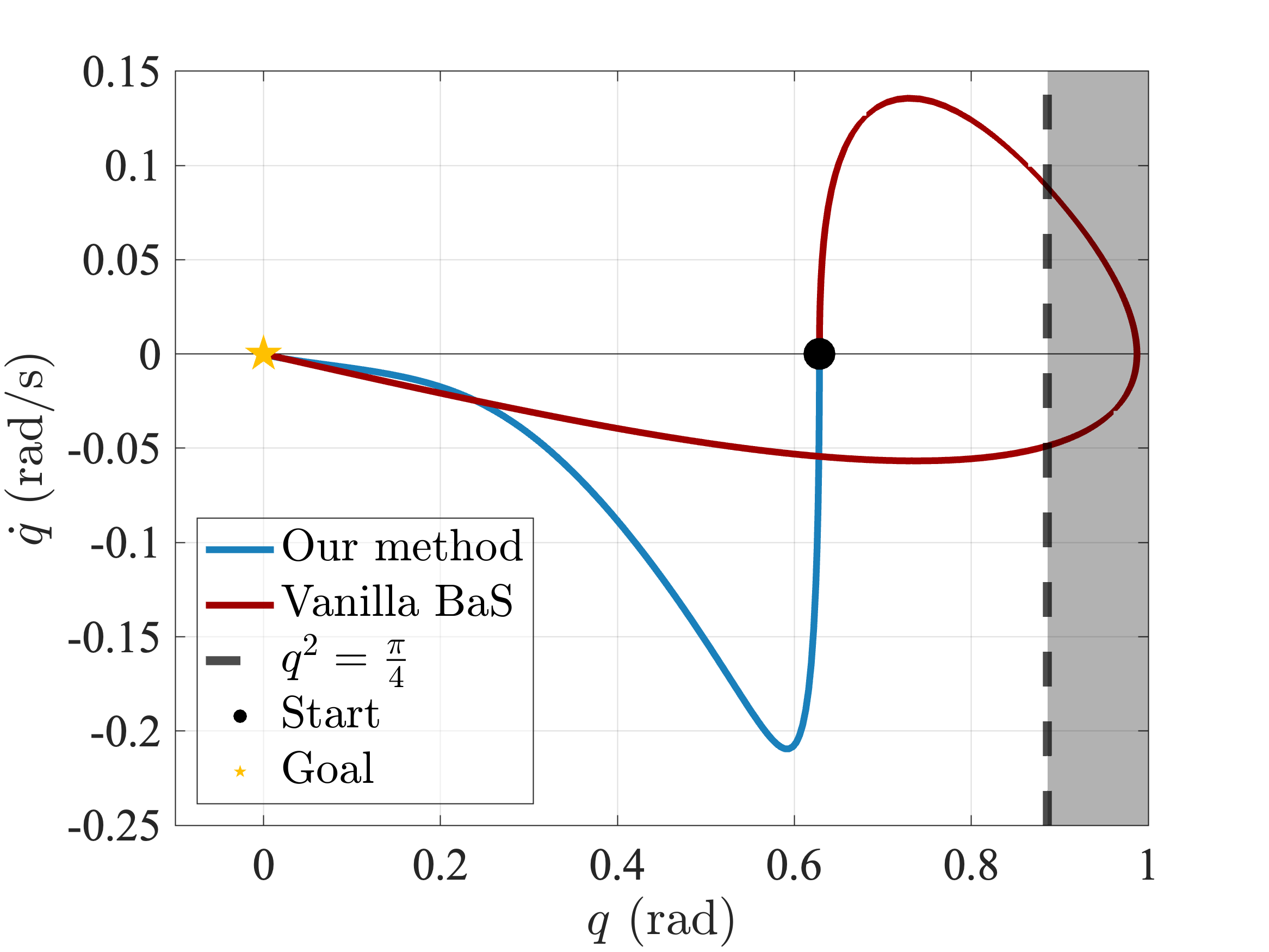}
    \caption{The proposed approach safely stabilizes the pendulum at $q=\SI[per-mode = symbol]{0}{\rad}$. Vanilla BaS cannot handle the situation due to the uncertain $g$ and $b$.}
    \label{fig:pendulum_sim}
    \vspace{-5mm}
\end{figure}

\subsection{Adaptive cruise control}
The system dynamics are given by:
\begin{align}
{\begin{bmatrix} \dot{v} \\ \dot{D} \end{bmatrix}} = {\begin{bmatrix} 0 \\ v_L - v \end{bmatrix} + \begin{bmatrix} \frac{1}{m} \\ 0 \end{bmatrix} u} {- \frac{1}{m} \begin{bmatrix} 1 & v & v^2 \\ 0 & 0 & 0 \end{bmatrix}} {\begin{bmatrix} f_0 \\ f_1 \\ f_2 \end{bmatrix}},
\label{eq:acc_sys}
\end{align}
where $v$ and $D$ are the ego vehicle's velocity and its distance from a leading vehicle moving at speed $v_L$. The mass of the ego vehicle is $m$, and $f_0$, $f_1$, and $f_2$ are unknown parameters related to rolling frictional force. The control input $u$ represents the ego vehicle’s acceleration.

For our example, we use the following parameters: $m = \SI[per-mode = symbol]{1650}{\kilogram}$, $v_L = \SI[per-mode = symbol]{14}{\meter\per\second}$, and the unknowns $f_0 = \SI[per-mode = symbol]{0.1}{\newton}$, $f_1 = \SI[per-mode = symbol]{5}{\newton\second\per\meter}$, and $f_2 =  \SI[per-mode = symbol]{0.25}{\newton\square\second\per\meter}$. In \autoref{fig:acc_sim}, we show an example where the goal for the vehicle is to maintain a speed of $v_{des} = \SI[per-mode = symbol]{24}{\meter\per\second}$ while satisfying the safety constraint $D \geq 1.8v$. We compare our method with raCBF and BaS without adaptation and we see that our method performs better than raCBF (i.e., approaches the desired speed and the boundary of the safe set more closely) while BaS without adaptation fails to satisfy the safety condition.
\begin{figure}[ht]
\vspace{-4mm}
    \centering
    \includegraphics[width=\linewidth]{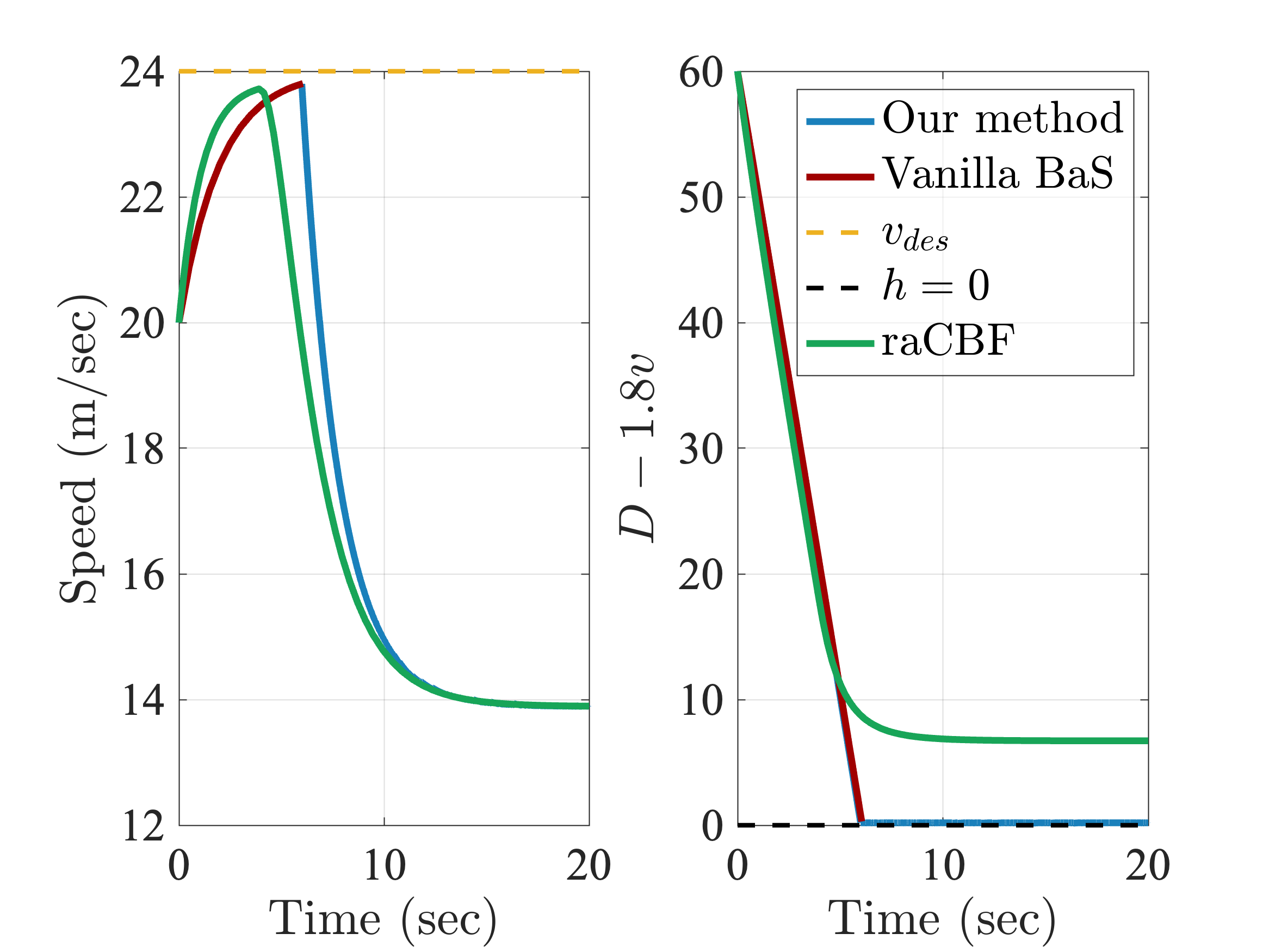}
    \caption{Our method increases the speed until the safe set's boundary set is approached and then it decreases the speed. Notice that it can approach the boundary with a small gain ($\Gamma = I$), while raCBF is more conservative even with a very high gain ($\Gamma = 200I$).}
    \label{fig:acc_sim}
    \vspace{-3mm}
\end{figure}

\section{Conclusions} \label{sec:conclusions}
In conclusion, we introduced a framework that effectively integrates adaptive control techniques with barrier states. By embedding barrier states directly within the control law, we successfully address the challenge of maintaining system safety while achieving desired performance in the presence of uncertainties. The results presented here contribute to advancing the state of the art in adaptive control, safety-critical systems, and multi-objective control frameworks.

Future work will focus on extending this framework to different adaptive control techniques and exploring its application in real-world safety-critical environments, where both adaptability and safety must be maintained.

\printbibliography

\end{document}